\documentclass[12pt, draftcls, onecolumn, twoside]{IEEEtran}
\usepackage{qtree}
\usepackage{bbm}
\usepackage{amsmath}
\usepackage{amssymb}
\usepackage{amsfonts}
\usepackage{epsfig}
\usepackage{calc}
\usepackage{color}
\usepackage[all]{xy}
\usepackage{bm}
\usepackage{algorithmicx}
\usepackage{algpseudocode}
\usepackage{algorithm}
\usepackage{enumerate}
\usepackage{pdfsync}
\usepackage{framed}
\usepackage{caption}
\usepackage{amsthm}
\usepackage{blkarray}
\usepackage{subcaption}
\usepackage{enumitem}
\usepackage{hyperref}
\usepackage{mathrsfs}
\usepackage{caption}
\usepackage{float}
\newtheorem{defn}{Definition}[section]

\newtheorem{thm}{Theorem}[section]
\newtheorem{cor}[thm]{Corollary}
\newtheorem{prop}{Proposition}

\newtheorem{lem}[thm]{Lemma}
\newtheorem{conj}[thm]{Conjecture}
\newtheorem{constr}[thm]{Construction}
\usepackage[normalem]{ulem}
\newtheorem{proposition}[thm]{Proposition}
\newtheorem{remark}{Remark}[section]
\newtheorem{example}{Example}[section]

\newcounter{definition}[section]

\newcommand{\bit}{\begin{itemize}}
\newcommand{\eit}{\end{itemize}}
\newcommand{\bcor}{\begin{cor}}
\newcommand{\ecor}{\end{cor}}
\newcommand{\beq}{\begin{equation}}
\newcommand{\eeq}{\end{equation}}
\newcommand{\beqn}{\begin{equation*}}
\newcommand{\eeqn}{\end{equation*}}
\newcommand{\bea}{\begin{eqnarray}}
\newcommand{\eea}{\end{eqnarray}}
\newcommand{\bean}{\begin{eqnarray*}}
\newcommand{\eean}{\end{eqnarray*}}
\newcommand{\ben}{\begin{enumerate}}
\newcommand{\een}{\end{enumerate}}
\newcommand{\bdefn}{\begin{defn}}
\newcommand{\edefn}{\end{defn}}
\newcommand{\bnote}{\begin{remark}}
\newcommand{\enote}{\end{remark}}
\newcommand{\bprop}{\begin{prop}}
\newcommand{\eprop}{\end{prop}}
\newcommand{\blem}{\begin{lem}}
\newcommand{\elem}{\end{lem}}
\newcommand{\bthm}{\begin{thm}}
\newcommand{\ethm}{\end{thm}}
\newcommand{\bconj}{\begin{conj}}
\newcommand{\econj}{\end{conj}}
\newcommand{\bconstr}{\begin{constr}}
\newcommand{\econstr}{\end{constr}}
\newcommand{\bpf}{\begin{proof}}
\newcommand{\epf}{\end{proof}}

\begin{document}

\title{Rate-Optimal Streaming Codes for Channels with Burst and Isolated Erasures}
\author{\IEEEauthorblockN{M. Nikhil Krishnan, {\em Student Member, IEEE}, P. Vijay Kumar, {\em Fellow, IEEE}}\\
	\IEEEauthorblockA{Electrical Communication Engineering, Indian Institute of Science, Bangalore - 560012 \\
		email: \{nikhilkrishnan.m, pvk1729\}@gmail.com}
	\thanks{P. Vijay Kumar is also a Visiting Professor at the University of Southern California. This research is supported in part by the National Science Foundation under Grant 1421848 and in part by an India-Israel UGC-ISF joint research program grant. M. Nikhil Krishnan would like to acknowledge the support of Visvesvaraya PhD Scheme for Electronics \& IT awarded by Department of Electronics and Information Technology, Government of India.}
%
%
%
%
}	

	\maketitle

\begin{abstract}
Recovery of data packets from packet erasures in a timely manner is critical for many streaming applications. An early paper by Martinian and Sundberg introduced a framework for streaming codes and designed rate-optimal codes that permit delay-constrained recovery from an erasure burst of length up to $B$. A recent work by Badr et al. extended this result and introduced a sliding-window channel model $\mathcal{C}(N,B,W)$. Under this model, in a sliding-window of width $W$, one of the following erasure patterns are possible (i) a burst of length at most $B$ or (ii) at most $N$ (possibly non-contiguous) arbitrary erasures. Badr et al. obtained a rate upper bound for streaming codes that can recover with a time delay $T$, from any erasure patterns permissible under the $\mathcal{C}(N,B,W)$ model. However, constructions matching the bound were absent, except for a few parameter sets. In this paper, we present an explicit family of codes that achieves the rate upper bound for all feasible parameters $N$, $B$, $W$ and $T$.
\end{abstract}

\section{Introduction}\label{sec:intro}
In many multimedia streaming applications, where packet losses are a norm, fast recovery of lost packets is often desirable. Conventional coding schemes like linear block codes, while helpful in combating erasures, can have increased delays due to buffering of data. In \cite{MartSunTIT04}, Martinian and Sundberg  introduce a framework for convolutional codes with decoding delay as an explicit parameter. We refer to these codes as {\it streaming codes}. The setting considered is as follows:

At each time instance $t \in \{0,1,2,\ldots\}$, an encoder $E$ receives a message packet $\mathbf{s}[t]$ from a source stream. Each message packet $\mathbf{s}[t]$ consists of $k$ symbols drawn from a finite field $\mathbb{F}_q$ of size $q$. In other words, $\mathbf{s}[t] \triangleq(s_0[t]\ s_1[t]\ \ldots\ s_{k-1}[t])^\intercal\in \mathbb{F}_q^k $. $E$ is a convolutional encoder and emits a coded packet $\mathbf{x}[t]$ at time $t$, where $\mathbf{x}[t]\triangleq(x_0[t]\ x_1[t]\ \ldots\ x_{n-1}[t])^\intercal\in\mathbb{F}_q^n$. $E$ is also assumed to be causal, wherein each coded packet $\mathbf{x}[t]$ is a function  of the message packets until time $t$, i.e., $\mathbf{s}[0],\mathbf{s}[1],\ldots,\mathbf{s}[t]$. Between the encoder-decoder pair, there exists a channel which introduces erasures at packet level. Let $\mathbf{y}[t]$ denote the packet received at the decoder end. We have:
\bean
\mathbf{y}[t] & = & \left\{ \begin{array}{rl} * & \text{if $\mathbf{x}[t]$ is erased}, \\
	\mathbf{x}[t] & \text{otherwise}. \end{array} \right.
\eean

Let $\hat{\mathbf{s}}[t]$ denote the decoded packet corresponding to ${\mathbf{s}}[t]$. The decoder is delay-constrained with a delay parameter $T$. i.e., each decoded message packet $\hat{\mathbf{s}}[t]$ is obtained as a function of received coded packets $\{\mathbf{y}[0],\mathbf{y}[1],\ldots,\mathbf{y}[t+T]\}$. Note that some of these packets can possibly be erased by the channel.  The {\it rate} $R$ of the code is naturally defined as $\frac{k}{n}$.  

 In \cite{MartSunTIT04}, the authors consider a channel which introduces a burst of length at most $B$. Together with the delay-constraint $T$, the streaming codes designed for this channel can in fact tolerate multiple erasure bursts, each of length at most $B$, with a guard space of at least $T$ between consecutive bursts. However, as the codes presented in \cite{MartSunTIT04} are tuned for burst erasures, they are sensitive to isolated erasures. In a recent work by Badr et al. \cite{BadrPatilKhistiTIT17}, the authors introduce a sliding-window based channel model which accounts for burst erasures and isolated erasures. Under the sliding-window channel model, in any sliding-window of width $W$, the channel will have one of the following erasure patterns; (a) a burst erasure of length at most $B$ {\it or} (b) up to $N$ erasures at arbitrary locations within the sliding-window. The channel is denoted by $\mathcal{C}(N,B,W)$. Clearly, $N\leq B$ as burst erasures are a special case of arbitrary erasures. Also, $T\geq B$ as otherwise non-zero rates are not feasible for a streaming code. Taken together with the delay-constraint $T$, the $\mathcal{C}(N,B,W)$ model specializes to the burst-erasure-only model in \cite{MartSunTIT04}, if one chooses $W=T+1$ and $N=1$. Furthermore, a rate upper bound on codes for this channel has also been derived in \cite{BadrPatilKhistiTIT17}. There are several other follow-up works that study various  low-delay communication schemes like \cite{LiKhistiGirodMultipleBursts}, \cite{TekinHYJ12}, \cite{LeongQH13}, \cite{AdlerC17} and references therein.

{\it In an independent, concurrent work \cite{FongStreamingCodes}, the authors prove the existence of rate-optimal streaming codes for all parameters. }

{\it Our Results:} In this work, we prove the tightness of the rate upper bound derived in \cite{BadrPatilKhistiTIT17} for codes that allow both burst erasures and isolated erasures, with a delay-constraint of $T$. We prove this by constructing explicit codes that meet the bound with equality. The family of codes that achieve the bound are based on linearized polynomials, and require a field-size exponential in $T$. For various range of parameters, we obtain rate-optimal codes that require lower field-size, which is of $O(T^2)$. In some cases, even linear and binary field-sizes suffice. We obtain rate-optimal streaming codes by reducing the problem to the design of linear block codes with certain properties. For the corner case of $B=N$, these codes specialize to MDS codes.

\section{Preliminaries}\label{sec:prelims}
%
%
%

\subsection{Streaming Capacity \cite{BadrPatilKhistiTIT17}}
In the context of streaming codes with decoding delay-constraint $T$, a rate $R$ is said to be achievable over $\mathcal{C}(N,B,W)$, if there exists a streaming code with delay parameter $T$ and rate $R$ that tolerates all the erasure patterns permitted by $\mathcal{C}(N,B,W)$. The supremum of all such rates is termed as the streaming capacity. In \cite{BadrPatilKhistiTIT17}, the authors obtain the following upper bound for $R$:

\begin{equation}\label{eq:rate_upperbound}
{R\leq \frac{T_{\text{eff}}-N+1}{B+T_{\text{eff}}-N+1}},
\end{equation}
where $T_{\text{eff}}=\min\{T,W-1\}$. We shall refer to $T_{\text{eff}}$ as the {\it effective delay}. Achievability of this rate bound is not known in general, except for a limited set of parameters. For the burst-alone case, i.e., $N=1$, Maximally Short (MS) codes (\cite{MartSunTIT04,MartTrotISIT07}) are known to meet the bound \eqref{eq:rate_upperbound}. For the other extreme case of $N=B$, Strongly-MDS convolutional codes \cite{GluesRoseSTIT06} are shown to be optimal in \cite{BadrPatilKhistiTIT17}. \cite{BadrKTAInfocom13} provides a family of codes with $R=0.5$ that meets \eqref{eq:rate_upperbound}. \cite{BadrPatilKhistiTIT17} shows the existence of near-optimal codes for all feasible parameters $B,N,T_{\text{eff}}$ that have a guaranteed rate of at least $\frac{T_{\text{eff}}-N}{B+T_{\text{eff}}-N}$. 

Let $[u,v]\triangleq\{u,u+1,\ldots,v\}$. We have the following lemma.

\begin{lem}\label{lem:sliding_window_to_teff}
	If there is a  streaming code $\mathscr{C}_{\text{str}}$ which permits recovery of any packet $\mathbf{x}[t]$ from erasure with a delay of at most $T_{\text{eff}}$, even in presence of either (i) a burst erasure (involving time $t$) of length at most $B$ or (ii) at most $N$ isolated erasures (again, including time $t$), then $\mathscr{C}_{\text{str}}$ tolerates all the erasure patterns permitted by $\mathcal{C}(N,B,W)$ with a delay-constraint $T$. 
\end{lem}

\begin{proof}
	Consider any erasure pattern that arises from the $\mathcal{C}(N,B,W)$ model. Let $j\in \{1,2,\ldots\}, J\in\{2,3,\ldots\}$. Denote by $t_j$, the coordinate at which the $j^\text{th}$ erasure happened. Let the induction assumption be that, for all $j\leq (J-1)$,  $\mathbf{x}[t_j]$ is recovered by the `time' $(t_j+T_\text{eff})$, i.e., with a delay $T_\text{eff}\leq T$. Hence during the decoding of $\mathbf{x}[t_J]$, we can assume that, strictly before time $t_J$, there are no erasures occurring among the coordinates $[0,t_J+T_\text{eff}]$.
	
	Consider the window of width $W$ consisting of coordinates $[t_J,t_J+W-1]$. The coded packet $\mathbf{x}[t_J]$ is erased by assumption. By definition of the sliding window channel model, in the worst case, there will be either (i) $(B-1)$ erasures at coordinates $[t_J+1,t_J+B-1]$ or (ii) $(N-1)$ erasures within $[t_J+1,t_J+W-1]$. As $T_\text{eff}\leq (W-1)$, when restricted to the set of coordinates $[0,t_J+T_\text{eff}]$, $\mathscr{C}_\text{str}$ will be observing either a burst of length $\leq B$ including coordinate $t_J$ or $\leq N$ arbitrary erasures involving $t_J$ (note that the induction assumption removes all erasures occurring before time $t_J$). In either case, $\mathscr{C}_\text{str}$ will be able to recover $\mathbf{x}[t_J]$ by time $(t_J+T_\text{eff})$, with delay $T_\text{eff}\leq T$. 
	
	The base case of the induction can be shown to be true by taking $J=1$. Thus we have shown that for all $j\geq 1$, $\mathbf{x}[t_j]$ can recover from any erasure pattern permitted by $\mathcal{C}(N,B,W)$, with a delay of $T_\text{eff}\leq T$.
\end{proof}

Hence in order to construct streaming codes for $\mathcal{C}(N,B,W)$ with a delay parameter of $T$, it is enough to focus on streaming codes that can handle an erasure burst of length $B$ or $N$ isolated erasures, with a delay of at most $T_{\text{eff}}$. We drop the subscript from the notation $T_{\text{eff}}$ for brevity and refer to $T$ as the (effective) delay parameter. 

\subsection{Diagonal Interleaving}

Diagonal interleaving is a known technique that enables one to convert a block code into a convolutional code. It has been used in the context of streaming codes in works like \cite{MartSunTIT04}, \cite{MartTrotISIT07} to reduce the problem of designing rate-optimal streaming codes for burst-erasure-only case ($N=1$) to that of designing block codes having appropriate features. Consider an $[n,k]$ linear block code $\mathscr{C}$, where $n$, $k$ denote the code-length and dimension, respectively. We summarize the diagonal interleaving using Fig. \ref{fig:diag_interleaving}. In the figure, we consider a systematic encoder for $\mathscr{C}$ and hence take $x_i[t]=s_i[t]$ for $0\leq i\leq k-1$. In order to stress the point that the last $(n-k)$ symbols of a coded packet are parity symbols at any time $t$, we use the notation $p_j[t]\triangleq x_{k+j}[t]$, for $0\leq j\leq n-k-1$. As an example, consider the following parameters for $\mathscr{C}$; $n=5$, $k=3$ . The diagonal interleaving technique will result in a convolutional code as shown in Fig. \ref{fig:diag_interleaving_example}.

\begin{figure}[ht]
	\centering
	\captionsetup{justification=centering}
	\includegraphics[scale=0.5]{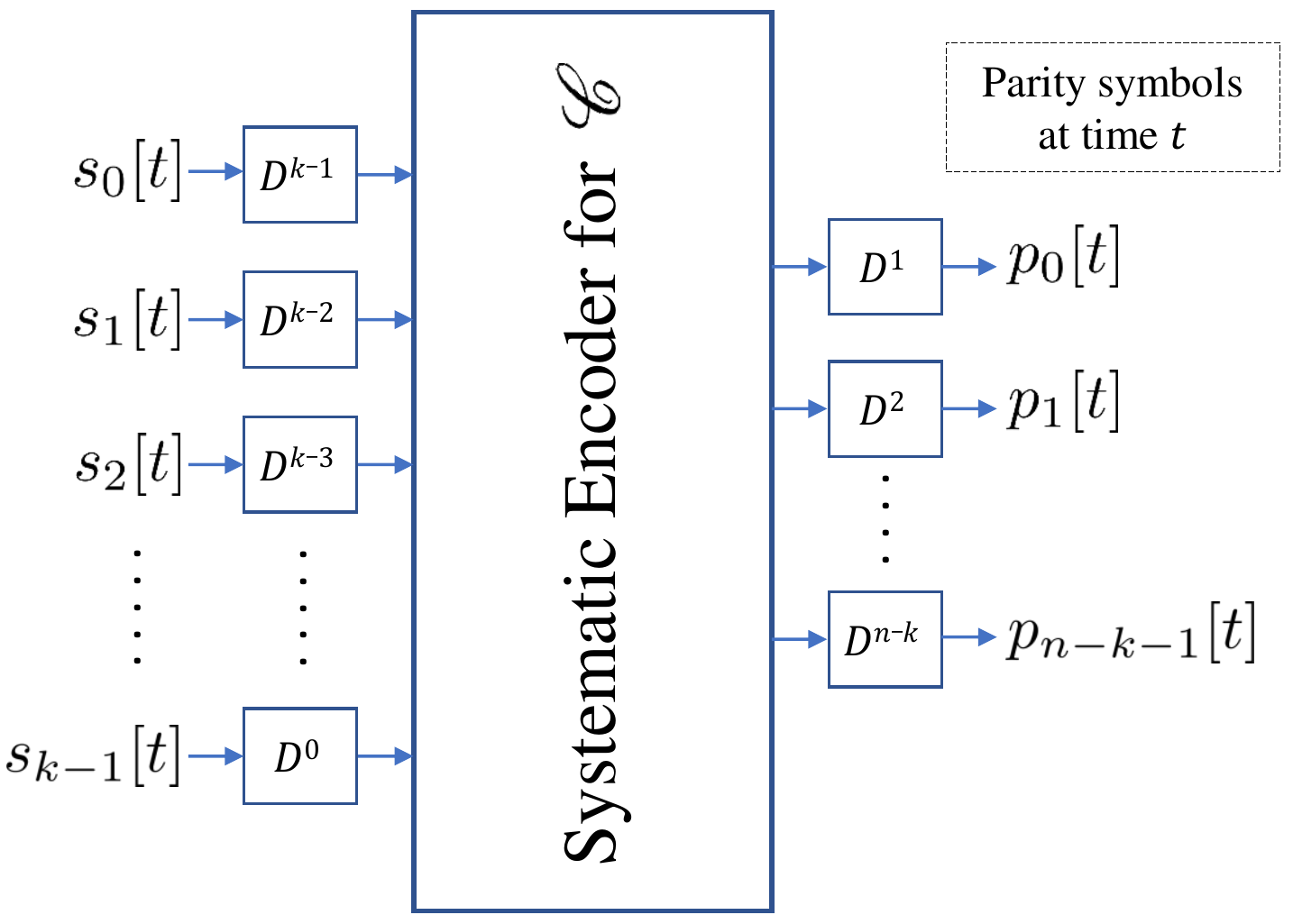}
	\caption{A depiction of the diagonal interleaving technique. Here $D^j$ denotes a delay of $j$ time units.}
	\label{fig:diag_interleaving}
\end{figure}

\begin{figure}[ht]
	\centering
	\captionsetup{justification=centering}
	\includegraphics[scale=0.5]{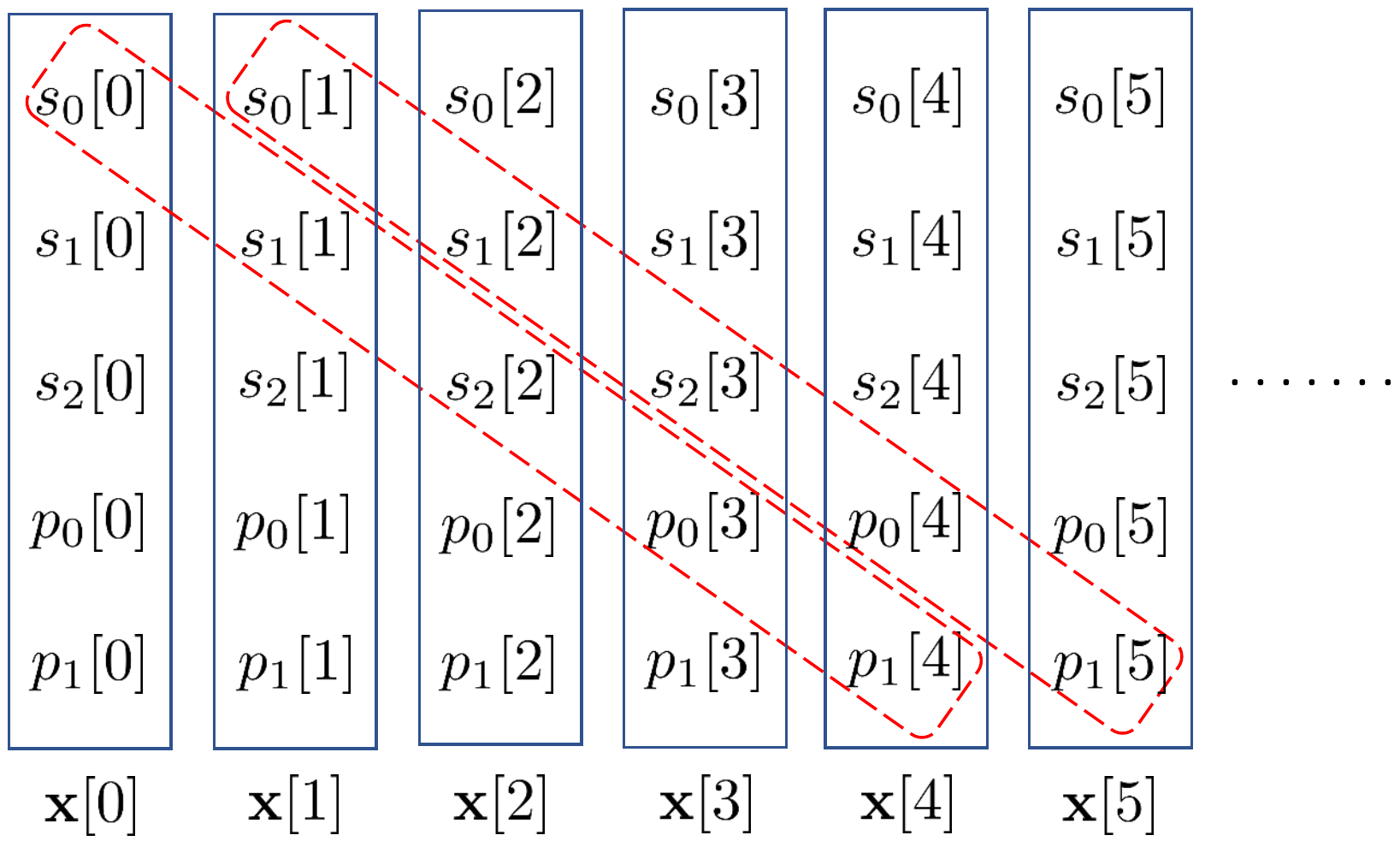}
	\caption{An illustration of the diagonal interleaving technique for a $[5,3]$ code. Each diagonal of the form ($s_0[t],s_1[t+1],s_2[t+2],p_0[t+3],p_1[t+4])$ is a codeword in $\mathscr{C}$, where $t\geq0$. For notational consistency, define $s_i[t]=p_i[t]\triangleq0$ for $t<0$.}
	\label{fig:diag_interleaving_example}
\end{figure}

	Consider an $[n,k]$ linear block code $\mathscr{C}$ over $\mathbb{F}_q$. Let $\mathbf{c}\triangleq(c_0,c_1,c_2, \ldots, c_{n-1})$ denote a codeword in $\mathscr{C}$.

	\begin{defn}
		For $i\in[0,n-1]$ and $\mathcal{I}\subseteq[0,n-1]$, the $i^{\text{th}}$ coordinate of $\mathscr{C}$ is said to be {\it recoverable} from ( the set of coordinates ) $\mathcal{I}$, if there exist $\lambda_j$'s chosen from $\mathbb{F}_q$ such that:
		\begin{equation*}
		c_i=\sum_{j\in\mathcal{I}}\lambda_jc_j \ \ \forall \mathbf{c}\in\mathscr{C}.
		\end{equation*}
	\end{defn}
	If $i\in\mathcal{I}$, it is straightforward to see that $i^{\text{th}}$ coordinate is recoverable from $\mathcal{I}$. For parameters $N$, $B$ as in the channel model $\mathcal{C}(N,B,W)$ and the delay parameter $T$, let $\Delta_i\triangleq\min\{i+T,n-1\}$.

\begin{defn}\label{defn:blockcode_for_sliding_window}
$\mathscr{C}$ is said to be {\it conforming to} $\mathcal{C}(N,B,W)$ with delay-constraint $T$, if both the following statements hold for all $i\in [0,n-1]$:
\bit
\item $i^{\text{th}}$  coordinate of $\mathscr{C}$ is recoverable from $[0,\Delta_i]\setminus\mathcal{N}$, for all $\mathcal{N}\subseteq[0,n-1]$ such that $|\mathcal{N}|\leq N$.
\item $i^{\text{th}}$  coordinate of $\mathscr{C}$ is recoverable from $[0,\Delta_i]\setminus \mathcal{B}$, where $\mathcal{B}=[u,v]$, for all $0\leq u\leq v\leq n-1$ such that $(v-u+1)\leq B$.
\eit
\end{defn}

One can show without much difficulty that, if there is an $[n,k]$ linear block code $\mathscr{C}$ which conforms to  $\mathcal{C}(N,B,W)$ with delay-constraint $T$, by applying diagonal interleaving, it will result in a streaming code $\mathscr{C}_{\text{str}}$ that can tolerate an erasure burst of length $B$ or $N$ isolated erasures, with delay-constraint $T$. By Lemma \ref{lem:sliding_window_to_teff}, $\mathscr{C}_{\text{str}}$ yields an achievable rate $\frac{k}{n}$ over $\mathcal{C}(N,B,W)$  with effective delay-constraint $T$. Hence from here onwards we restrict our attention to designing linear block codes that conform to $\mathcal{C}(N,B,W)$ with (effective) delay-constraint $T$ and have rate $R=k/n$ that meets \eqref{eq:rate_upperbound} with equality. We refer to these codes as rate-optimal linear block codes with delay-constraint $T$. 

Let $T=aB+\delta$, where $a\geq 0$ and $1\leq \delta\leq B$. If $a=0$, $\delta=B$ (as $B\leq T$).

\section{Rate-Optimal Linear Block Codes with Delay Constraint $T$ for the Case $\delta\geq (B-N)$}

\subsection{Construction-$A$}
Given the channel $\mathcal{C}(N,B,W)$ and effective delay constraint $T$, consider the parameters $n=(B+T-N+1)$ and $k=(T-N+1)$ for the linear block code $\mathscr{C}$ to be constructed over $\mathbb{F}_{q^2}$. Let $\delta\geq (B-N)$ and $\mathbf{m}\triangleq[m_0\ m_1\ \ldots\ m_{k-1}]$ be the $k$-length message vector to be mapped to the codeword $\mathbf{c}\triangleq(c_0, c_1,\ldots,c_{n-1})\in\mathscr{C}$. Consider an $[n_{\text{MDS}}=T+1,k_{\text{MDS}}=k,d_{\text{min}}=N+1]$ MDS code $\mathscr{C}_{\text{MDS}}$ over $\mathbb{F}_q\subseteq\mathbb{F}_{q^2}$, with a generator matrix $\mathbf{G}_{\text{MDS}}$ having all the entries from $\mathbb{F}_q\subseteq \mathbb{F}_q^2$. Let 
$\mathbf{c}'=[c'_0\ c'_1 \ \ldots \ c'_T]=\mathbf{m}\mathbf{G}_{\text{MDS}}$. Assign these $(T+1)$ code-symbols $\{c'_j\}_{j=0}^T$ to the coordinates $\{0,1,\ldots,T-1\}\cup\{n-1\}$ of $\mathbf{c}$. i.e., $c_i=c'_i$ for $0\leq i\leq T-1$ and $c_{n-1}=c'_{T}$. For the $(B-N)$ coordinates $\{c_{T+j} :0\leq j\leq(B-N-1)\}$, take $c_{T+j}=\underbrace{\alpha c_j+c_{j+B}+\ldots+c_{j+aB}}_\text{$(a+1)$ terms}$. Note that, since $j\leq (B-N-1)\leq \delta-1$, all the $(a+1)$ code-symbols taking part in the check-sum belong to the set of coordinates $[0,T-1]$. Here $\alpha$ can be any field element from the set $\mathbb{F}_{q^2}\setminus\mathbb{F}_q$.
Existence of $\mathbf{G}_{\text{MDS}}$ is guaranteed when $q\geq (T+1)$. Hence $\mathscr{C}$ can be constructed with a field-size of $O(T^2)$.

Let $N=2, B=4, T=10$. Hence $a=2, \delta=2$. The dimension of $\mathscr{C}$, $k=T-N+1=9$ and length, $n=B+T-N+1=13$. This example code is illustrated in Fig. \ref{fig:example_special_case}.
\begin{figure*}[ht]
	\centering
	\captionsetup{justification=centering}
	\includegraphics[scale=0.6]{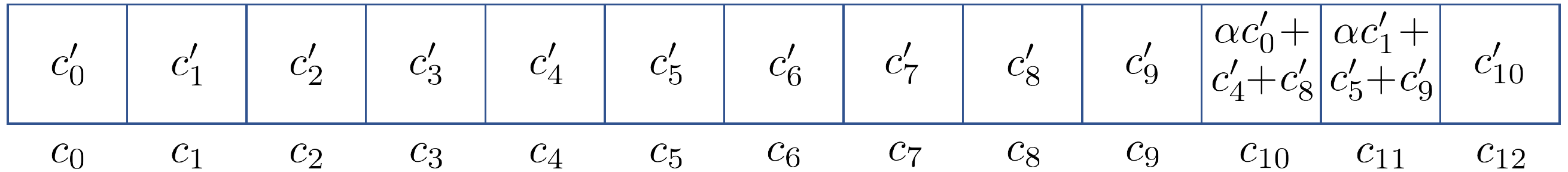}
	\caption{An illustration of construction-A for the parameters: $N=2, B=4, T=10$. Here the set $\{c'_i\}_{i=0}^{10}$ corresponds to code-symbols coming from an $[n_{\text{MDS}}=11,k_{\text{MDS}}=9]$ MDS code.}
	\label{fig:example_special_case}
\end{figure*}
\begin{proposition}\label{prop:del_gt_B_minus_N}
	Construction-$A$ gives a rate-optimal family of linear block codes with delay constraint $T$.
\end{proposition}
\begin{proof}
	{\it (Burst erasure recovery with delay constraint $T$)} Consider an erasure burst of length $B$. Let $\epsilon$ indicate the number of erasures occurred among the coordinates in $[T,T+B-N-1]$, as a part of the erasure burst. Thus, the set of non-erased symbols among the coordinates $[T,T+B-N-1]$ is given by $\{c_{i+T}:\epsilon\leq i\leq (B-N-1)\}$, where $c_{i+T}\triangleq\alpha c_i+c_{i+B}+\ldots+c_{i+aB}$. As $c_{i+T}$ is assumed to be a non-erased symbol, the burst of length $B$ ends at some coordinate $\leq (T+i-1)$. Note that, at most one code-symbol in the sum given by each $c_{i+T}$ can be a part of the burst erasure, as the code-symbols in the sum are chosen to be $B$ apart. Suppose that none of the code-symbols  $\{c_i,c_{i+B},\ldots,c_{i+aB}\}$, which constitute the sum given by $c_{i+T}$, are part of the erasure burst. As the erasure burst is of length $B$, the burst can start only on or after the coordinate $(i+aB+1)$. This means that the burst will end at a coordinate $\geq i+aB+B\geq T+i$ (as $\delta \leq B$). This is a contradiction, and hence each non-erased coordinate in the range $[T+\epsilon,T+B-N-1]$ will have lost precisely one code-symbol from the check-sum it contains. Each of these code-symbols can be recovered from the check-sums. Thus there will be $(B-N-\epsilon)$ coordinates recovered from the $(B-\epsilon)$ erasures occurred across coordinates $[0,T-1]\cup\{T-N+B\}$. The remaining number of erasures is $B-\epsilon-(B-N-\epsilon)=N$. As $\mathscr{C}$ restricted to the coordinates $[0,T-1]\cup\{T-N+B\}$ is an $N$ erasure correcting MDS code, the remaining erasures can be corrected as well.
	
	Furthermore, we need to show that all the erased coordinates in the erasure burst can be recovered with a delay of at most $T$. Towards this, we need to consider only erasure bursts involving at least one of the coordinates in $[0,B-N-1]$. For all the remaining coordinates, the delay constraint is trivially met as the last coordinate of $\mathscr{C}$ is $(T+B-N)$. Suppose $c_i:i\in[0,B-N-1]$ is erased (as part of a burst erasure of length $B$). Clearly, none of the code-symbols  $\{c_{i+B},c_{i+2B},\ldots,c_{i+aB},c_{T+i}\}$ will be part of the erasure burst and hence $c_i$ can be recovered with a delay of $T$.
	
	{\it (Recovery from $N$ arbitrary erasures with delay constraint $T$)} $\mathscr{C}$ can recover from any $N$ arbitrary erasures, as it is formed by adding $(B-N)$ dependent coordinates to an $N$ erasure correcting MDS code. As in the case of recovery from burst erasures, in order to show that delay constraints are met during the recovery, one needs to consider only those $N$-erasure patterns involving at least one of the code-symbols from the set of coordinates $[0,B-N-1]$. Suppose $c_i$ is erased for some $i\in[0,B-N-1]$. Among the set of $T$ coordinates given by $[0,i-1]\cup[i,T-1]\cup\{T+i\}$, there can possibly be $(N-1)$ more erasures. 
	
	If the coordinate $(T+i)$ is among the erased coordinates, there will be possibly $(N-1)$ erasures among the coordinates $[0,T-1]$, inclusive of the coordinate $i$. As these coordinates are part of an MDS code of dimension $k=T-N+1$, the code-symbol $c_i$ can be recovered with a delay $<T$.   
	 
    Now consider the remaining case that $(T+i)$ is not an erased coordinate. Let $G_1$ be the systematic $k\times (T+1)$-generator matrix for the code obtained by puncturing $\mathscr{C}$ to the coordinates $[0,T-1]\cup \{T+i\}$. For $0\leq j\leq T-N$, the $j^\text{th}$ column of $G_1$ is taken as $e_j\triangleq[\ \underbrace{0\ \ldots\ 0}_\text{$j$}\ \underbrace{1 \ 0\ \ldots 0}_\text{$(k-j)$}\ ]^\intercal$. We remark that all the elements of the sub-matrix obtained by restricting $G_1$ to the coordinates corresponding to $[0,T-1]$, belong to $\mathbb{F}_q\subseteq \mathbb{F}_{q^2}$.  Let $G_2$ be the $(T-N)\times (T-1)$ sub-matrix obtained from $G_1$ after removing columns corresponding to the coordinates $\{i,(T+i)\}$ of $\mathscr{C}$ and the $i^\text{th}$ row. Note that $G_2$ corresponds to the MDS generator matrix of a code of length $(T-1)$ and dimension $(T-N)$. The column of $G_1$ corresponding to the coordinate $(T+i)$ will have the form $[\underbrace{*\ \ldots\ *}_\text{$i$}\ \underbrace{\beta\ *\ \ldots\ *}_\text{$(k-i)$}]^\intercal$, where $*$ indicates elements from $\mathbb{F}_q\subseteq \mathbb{F}_q^2$ and $\beta\in \mathbb{F}_q^2\setminus\mathbb{F}_q$. Let $G_1'$ be an $(T-N+1)\times (T-N+1)$ sub-matrix of $G_1$ with $(T-N)$ columns from $[0,i-1]\cup[i+1,T-1]$ and $(T+i)$. After expanding the determinant of $G_1'$ based on the column corresponding to the coordinate $(T+i)$, we have $det(G_1')=*+\beta*det(G_2')$, where $*$ denotes field elements from $\mathbb{F}_q\subseteq\mathbb{F}_q^2$ and $G_2'$ is some $(T-N)\times (T-N)$ sub-matrix of $G_2$. As $G_2$ is an MDS generator matrix with all the entries from $\mathbb{F}_q$, $det(G_2')\neq 0\in\mathbb{F}_q$ and hence $det(G_1')\neq0$. Thus, $c_i$ can be recovered with a delay of $T$.
\end{proof}
\begin{cor}{(Binary field-size)}
When $N=1$ and $\delta\geq (B-1)$, $\mathscr{C}$ can be over $\mathbb{F}_2$.	
\end{cor}
\begin{remark}
	When $B=N$, $\mathscr{C}$ reduces to an $[n=T+1,k=T-B+1]$ MDS code.
\end{remark}
\begin{remark}(Linear field-size for $N=(B-1)$)
For $N=(B-1)$, if the coordinates $0$ and $T$ are swapped in $\mathscr{C}$, with $\alpha\neq 0\in\mathbb{F}_q$, the resultant code will be rate-optimal with delay constraint $T$. In this case, field-size requirement will be $\geq (T+1)$.
\end{remark}
\section{Rate-Optimal Linear Block Codes with Delay Constraint $T$ for all Feasible Parameters: $\delta,B,N,T$}

\subsection{Linearized Polynomials}

A linearized polynomial  \cite{LidlFiniteFields} of $q$-degree $(k-1)$ is a polynomial of the form $f(x)=\sum_{i=0}^{k-1} m_ix^{q^i}$, where $m_i \in \mathbb{F}_{q^m}$ and $m_{k-1}\neq 0$. It satisfies the property that: $f(a_1x_1+a_2x_2)=a_1f(x_1)+a_2f(x_2)$, where $a_i\in \mathbb{F}_q$, $x_i\in \mathbb{F}_{q^m}$. Furthermore, a linearized polynomial of $q$-degree $(k-1)$ can be uniquely determined from evaluations at $k$ points $\{\theta_i\}_{i=0}^{k-1} \subseteq \mathbb{F}_{q^m}$, which are linearly independent over $\mathbb{F}_{q}$. Gabidulin codes \cite{Gab} are constructed based on linearized polynomial evaluations.

\subsection{Construction-B}

Given the channel $\mathcal{C}(N,B,W)$ and effective delay constraint $T$, consider the parameters $n=(B+T-N+1)$, $k=(T-N+1)$ for the linear block code $\mathscr{C}$ to be constructed over $\mathbb{F}_{q^m}$. The $k$ message symbols  $\{m_i\}_{i=0}^{k-1}$ will be taken as coefficients of a linearized polynomial with $q$-degree $(k-1)$ and evaluated at $(T+1)$ points $\{\theta_i\}_{i=0}^T$. Here $\theta_i\in \mathbb{F}_{q^m}$ and  $\{\theta_i\}_{i=0}^T$ is a collection of independent field elements over $\mathbb{F}_q\subseteq \mathbb{F}_{q^m}$. Let ${c'_0,c'_1,\ldots,c'_T}$ be the code-symbols (evaluations) thus obtained. Assign these code-symbols to the coordinates $\{0,1,\ldots,T-1\}\cup\{n-1\}$ of $\mathbf{c}$. i.e., $c_j=c'_j$ for $0\leq j\leq T-1$ and $c_{n-1}=c'_{T}$. Let $\delta'\triangleq \min\{\delta,(B-N)\}$. For coordinates $0\leq j\leq \delta'-1$, take $c_{T+j}=c_j+c_{j+B}+\ldots+c_{j+aB}$. For coordinates in the range $\delta'\leq j\leq (B-N)-1$, $c_{T+j}=c_j+c_{B+j}+\ldots+c_{(a-1)B+j}+\gamma_{j,0}c_{aB}+\gamma_{j,1}c_{aB+1}+\ldots+\gamma_{j,\delta-1}c_{aB+\delta-1}$. Here $\{\gamma_{i,j}\}$'s are selected in such a way that the $(B-N-\delta')\times\delta$ matrix of the form:
\begin{align*}
\boldsymbol{\Gamma} &\triangleq 
\begin{bmatrix}
\gamma_{\delta',0} & \gamma_{\delta',1} & \ldots & \gamma_{\delta',\delta-1}\\           
\gamma_{\delta'+1,0} & \gamma_{\delta'+1,1} & \ldots & \gamma_{\delta'+1,\delta-1} \\
\vdots & & & \vdots\\
\gamma_{B-N-1,0} & \gamma_{B-N-1,1} & \ldots & \gamma_{B-N-1,\delta-1}
\end{bmatrix}\label{eq:ConstrB_coeff},
\end{align*}
is Cauchy, with all the entries belonging to $\mathbb{F}_q\subseteq\mathbb{F}_{q^m}$. $\boldsymbol{\Gamma}$ can be constructed explicitly (for example, see \cite[Ch.~5]{Roth2006}) for $q\geq (\delta + B-N-\delta)=(B-N) $.  The evaluation points $\{\theta_i\}$ can be chosen as follows. Let $\hat{\beta}$ denote a primitive element of the extension field $\mathbb{F}_{q^{T+1}}$. A natural candidate for $\theta_i$, for $0\leq i\leq T$ will be $\theta_i=\hat{\beta}^{i-1}$. Hence construction-B is possible with a field-size $\geq q^{(T+1)}$, where $q\geq (B-N)$.

\begin{proposition}\label{prop:all_params}
	Construction-$B$ gives a rate-optimal family of linear block codes, for all feasible parameters $\delta,B,N,T$.
\end{proposition}

\begin{proof}
	
	{\it (Recovery from $N$ arbitrary erasures with delay constraint $T$)} Similar to the scenario in Proposition \ref{prop:del_gt_B_minus_N}, $\mathscr{C}$ is $N$-erasure correcting, as $\mathscr{C}$ is constructed by adding $(B-N)$ dependent coordinates to an $N$-erasure correcting MDS code (obtained from linearized polynomial evaluations). Hence we only need to show that delay requirements are met for arbitrary erasures involving at least one of the coordinates in $[0,(B-N-1)]$. For $i\in[0,(B-N-1)]$, let $i$ be an erased coordinate. Consider the set of coordinates $\mathcal{R}_i\triangleq[0,i-1]\cup[i+1,T-1]\cup\{T+i\}$. The corresponding code-symbols will be evaluations of a linearized polynomial at the evaluation points (field elements): $A_i\cup\{\theta_i+\sum_{\theta\in B_i}{\theta}\}$, where $A_i\triangleq\{\theta_0,\theta_1,\ldots,\theta_{i-1},\theta_{i+1},\ldots,\theta_{T-1}\}$, $B_i\triangleq \{\theta_i,\theta_{B+i},\ldots,\theta_{(a-1)B+i},\theta_{aB+i}\}$, if $0\leq i\leq \delta'-1$ or 
	$B_i\triangleq\{\theta_i,\theta_{B+i},\ldots,\theta_{(a-1)B+i},\gamma_{i,0}\theta_{aB},$
	$\gamma_{i,1}\theta_{aB+1},\ldots,$ 
	$\gamma_{i,\delta-1}\theta_{aB+\delta-1}\}$, if $\delta'\leq i\leq (B-N-1)$. 
	
	The set of $(T-1)$ evaluation points $A_i$ is an independent set over $\mathbb{F}_q$ by construction. Adding $\{\theta_i+\sum_{\theta\in B_i}{\theta}\}$ to $A_i$ will not add any dependency, as $\theta_i\notin A_i$. Thus, access to evaluations at any $k=(T-N+1)$ out of the $T$ points in $A_i\cup \{\theta_i+\sum_{j\in B_i}{\theta_j}\}$, will be sufficient to recover all the $k$ message symbols. This essentially implies that even if there are $(N-1)$ erasures among the coordinates in $\mathcal{R}_i$, $c_i$ can still be corrected with a delay of at most $T$.
	
	{\it (Burst erasure recovery with delay constraint $T$)} If $\delta\geq (B-N)$, proof exactly follows as in the case of Proposition \ref{prop:del_gt_B_minus_N}. For the remaining case of $\delta<(B-N)$, proof is moved to Appendix \ref{app:proof_all_params_case_del_lt_b_minu_n}.
\end{proof}

\begin{example}\normalfont\label{eg:general_case_example}
	Consider the parameters $B=7$, $N=2$, $\delta=3$, $T=10$. Hence for $\mathscr{C}$, we have the parameters $n=(B+T-N+1)=16$ and $k=(T-N+1)=9$. Also, $\delta=\delta'=3$. Let $\hat{\alpha}$ be a primitive element of $\mathbb{F}_{5}$ and
	\begin{align*}
	\boldsymbol{\Gamma} &= 
	\begin{bmatrix}
	1 & \hat{\alpha} & \hat{\alpha}^2\\           
	1 & \hat{\alpha}^2 & \hat{\alpha}^4
	\end{bmatrix}.
	\end{align*}
	Consider the set of $(T+1)=11$ field elements $\{\theta_i\}_{i=0}^{10}$, which forms an independent set over $\mathbb{F}_5$. For example, let $\hat{\beta}$ be a primitive element of $\mathbb{F}_{5^m}$, where $m\geq (T+1)=11$. Then, $\theta_i$ can be chosen as $\hat{\beta}^i$.
	
	Let $\{m_i\}_{i=0}^8$ be the set of $k=9$ message symbols drawn from $\mathbb{F}_{5^m}$. Consider the corresponding linearized polynomial $f(x)=\sum_{i=0}^{8} m_ix^{5^i}$. The code-symbols will be as follows: $c_i=f(\theta_i)$, for $0\leq i\leq (T-1)=9$, $c_{10}=c_0+c_7=f(\theta_0+\theta_7)$, $c_{11}=c_1+c_8=f(\theta_1+\theta_8)$, $c_{12}=c_2+c_9=f(\theta_2+\theta_9)$, $c_{13}=c_3+c_7+\hat{\alpha} c_8+\hat{\alpha}^2 c_9=f(\theta_3+\theta_7+\hat{\alpha} \theta_8+\hat{\alpha}^2 
	\theta_9)$, $c_{14}=c_4+c_7+\hat{\alpha}^2 c_8+\hat{\alpha}^4 c_9=f(\theta_4+\theta_7+\hat{\alpha}^2 \theta_8+\hat{\alpha}^4 
	\theta_9)$, $c_{15}=f(\theta_{10})$.
	
	Consider the recovery of $\mathscr{C}$ from $N=2$ arbitrary erasures, with a delay at most $10$. Clearly, $\mathscr{C}$ can tolerate $2$ erasures. This is because, it is constructed by adding $(B-N)=5$ additional parity checks (at coordinates $10$, $11$, $12$, $13$ and $14$) to a $2$-erasure correcting $[11,9]$-MDS code. As the last coordinate of $\mathscr{C}$ is $15$, delay constraint of $10$ will be trivially met during the recovery of any $c_j: j\geq 5$, which is a part of $2$ arbitrary erasures. Hence we need to consider only those $2$-erasure patterns which involve at least one coordinate from $[0,4]$. Suppose $c_1$ is erased. Consider the $T=10$ coordinates, $\mathcal{R}_1=\{0\}\cup[2,9]\cup\{11\}$. The corresponding $10$ evaluation points $\{\theta_0,\theta_2,\theta_3,\ldots,\theta_9,(\theta_1+\theta_8)\}$ form an independent set over $\mathbb{F}_5$. Therefore, even if one more coordinate is lost (since $N=2$, and coordinate $1$ is already assumed to be lost) from $\mathcal{R}_1$, there will still be $k=9$ independent evaluation points available from the non-erased coordinates in $\mathcal{R}_1$. As the last coordinate in $\mathcal{R}_1$ is $(T+1)=11$, $c_1$ can be recovered with a delay at most $10$. Similar arguments hold when $c_0$, $c_2$, $c_3$ or $c_4$ is part of a $2$-erasure pattern.
	
	For the burst erasure case, consider consecutive erasures of length $B=7$, which erases the set of coordinates $[u,v]$, where $0\leq u\leq v\leq (n-1)$ and $v=u+B-1$. Let $\epsilon$ denote the number of coordinates in $[u,v]$ that intersect $[T,T+B-N-1]=[10,14]$. Clearly, $0\leq \epsilon\leq (B-N)=5$.
	
	\bit
	\item  $\epsilon=5$: If the burst erasure is such that $\epsilon=(B-N)=5$, the $[11,9]$-MDS code across coordinates $[0,9]\cup\{15\}$ will have $B-(B-N)=N=2$ coordinates erased and hence all the erasures can be corrected. 
	
  \item	$\epsilon=0$: Each symbol from $[10,14]$ will recover one unique erased code-symbol. Let $[u,u+6]\subseteq[0,9]$ denote the $B=7$ coordinates involved in the burst erasure. Clearly, there exists $j_0,j_1,\ldots,j_{4} \in [u,u+6]$ such that $j_\ell \mod B=\ell$, for $0\leq \ell\leq B-N-1=4$. Code-symbol at coordinate $10+\ell$ will be recovering precisely the corresponding $j_\ell$. 	
  
\item 	$\epsilon=1$: The erased coordinates belong to $[4,10]$. As $c_{11}=c_1+c_8$ and $c_{12}=c_2+c_9$, these two code-symbols recover coordinates $8$ and $9$, respectively. From $c_{13}(=c_3+c_7+\hat{\alpha} c_8+\hat{\alpha}^2 c_9)$ and $c_{14}(=c_4+c_7+\hat{\alpha}^2 c_8+\hat{\alpha}^4)$, after removing the interference from the known code-symbols $\{c_3,c_8,c_9\}$, we obtain the sums $\hat{c}_{13}=c_7$ and $\hat{c}_{14}=c_4+c_7$, respectively. These two sums essentially yield the symbols $c_4$ and $c_7$.
	
\item 	$\epsilon=2$: The erased coordinates are in $[5,11]$. As $c_{12}=c_2+c_9$, it recovers coordinate $9$. After removing the known code-symbols $\{c_3,c_4,c_9\}$, the check-sums provided by $c_{13}$ and $c_{14}$ result in $\hat{c}_{13}=c_7+\hat{\alpha} c_8$ and $\hat{c}_{14}=c_7+\hat{\alpha}^2 c_8$, respectively. These sums would recover the code-symbols $c_7$ and $c_8$ (uses the Cauchy property of $\mathbf{\Gamma}$).
	
\item 	$\epsilon=3$: Here the erased coordinates are from the set $[6,12]$. After removing the interference from the known code-symbols $\{c_3,c_4\}$, $\hat{c}_{13}=c_7+\hat{\alpha} c_8+\hat{\alpha}^2 c_9=f(\theta_7+\hat{\alpha} \theta_8+\hat{\alpha}^2)$, $\hat{c}_{14}=c_7+\hat{\alpha}^2 c_8+\hat{\alpha}^4 c_9=f(\theta_7+\hat{\alpha}^2 \theta_8+\hat{\alpha}^4 
	\theta_9)$. From the Cauchy property of  $\mathbf{\Gamma}$, the evaluation points 
	$\{\theta_7+\hat{\alpha} \theta_8+\hat{\alpha}^2 
	\theta_9, \theta_7+\hat{\alpha}^2 \theta_8+\hat{\alpha}^4 
	\theta_9\}$ are both independent. These evaluation points along with $\{\theta_0,\theta_1,\ldots,\theta_5,\theta_{10}\}$ form an independent set of cardinality, $k=9$.
	
\item  $\epsilon=4$: The erased coordinates are in $[7,13]$. After removing the interference from the known code-symbols $\{c_4\}$, $\hat{c}_{14}=c_7+\hat{\alpha}^2 c_8+\hat{\alpha}^4 c_9=f(\theta_7+\hat{\alpha}^2 \theta_8+\hat{\alpha}^4 
	\theta_9)$. $\theta_7+\hat{\alpha}^2 \theta_8+\hat{\alpha}^4 
	\theta_9$ along with $\{\theta_0,\theta_1,\ldots,\theta_6,\theta_{10}\}$ result in an independent set of size $k=9$.
	\eit
Thus, for every $\epsilon$, the $(5-\epsilon)$ non-erased coordinates $[10+\epsilon,14]$ will provide one linearized polynomial evaluation each. These $(5-\epsilon)$ evaluations along with the  $11-(7-\epsilon)$ evaluations corresponding to the non-erased coordinates in $[0,10]\cup\{15\}$, result in linearized polynomial evaluations over $9$ independent evaluation points and hence all the erased symbols can be recovered.

In order to show that the delay constraint of $T=10$ is met during the recovery from burst erasures, we need to consider only the burst erasures involving at least one of the coordinates $[0,B-N-1]=[0,4]$. Suppose $c_i$, where $i\in[0,2]$, is part of a burst-erasure. The coordinates $\{7+i, 10+i\}$ will be non-erased, as they are at least $B=7$ apart from coordinate $i$. Hence using the check-sum provided by $c_{10+i}$, the code-symbol $c_i$ can be recovered, with a delay of $10$ time units. 

For $c_i$, where $i\in[3,4]$, which is part of a burst erasure, the symbol $c_{10+i}$ will be non-erased. There are three possible cases for $(i,\epsilon)$; $(3,0)$, $(4,0)$ and $(4,1)$. If $\epsilon=0$, the code-symbols $\{c_{7},c_8,c_{9}\}$ will be known by the time $(10+i)$. After removing the interference from the known symbols at coordinates $\{7,8,9\}$, $c_{10+i}$ will yield the code-symbol $c_i$. If $\epsilon=1$, we have already seen that, $c_4$ will be recovered at time $14$. 
\end{example}
\appendices
\section{Proof of Proposition \ref{prop:all_params}: Burst Erasures, Case: $\delta<(B-N)$}\label{app:proof_all_params_case_del_lt_b_minu_n}

	Consider the set of erased coordinates, $\mathcal{B}\triangleq[u,v]$, where $0\leq u\leq v\leq n-1$ and $v-u+1= B$. Let $|\mathcal{B}\cap[T,T+B-N-1]|=\epsilon$. 
	\ben
	\item	$\underline{\epsilon=(B-N)}$: The number of erasures across the coordinates $\{0,T-1\}\cup\{T-N+B\}$ will be $B-(B-N)=N$, which can be corrected. 
	
	\item $\underline{\epsilon=0}$: Each code-symbol from the set $\{c_T,c_{T+1},\ldots,c_{T+\delta-1}\}$ will correct one erased symbol (argument similar to that in the proof of Proposition \ref{prop:del_gt_B_minus_N}). This essentially means that all the code-symbols which are part of the check-sums given by $\{c_T,c_{T+1},\ldots,c_{T+\delta-1}\}$, are known by time $T+\delta-1$. In particular, we have $\{c_{aB},c_{aB+1},\ldots,c_{aB+\delta-1}\}$ known by time $T+\delta-1$, irrespective of whether some of them are erased or not. From each code-symbol $c_{T+j}\triangleq c_j+c_{B+j}+\ldots+c_{(a-1)B+j}+\gamma_{j,0}c_{aB}+\gamma_{j,1}c_{aB+1}+\ldots+\gamma_{j,\delta-1}c_{aB+\delta-1}$, for $\delta\leq j\leq B-N-1$, the interference from  $\{c_{aB},c_{aB+1},\ldots,c_{aB+\delta-1}\}$ can be canceled. Hence $c_{T+j}$'s can be thought of as simply, $\hat{c}_{T+j}=c_j+c_{B+j}+\ldots+c_{(a-1)B+j}$. For $0\leq i\leq B-N-1$, let $\ell_i$ be the unique coordinate $\in [u,v]$ such that $\ell_i\mod B=i$. The check-sum provided by the code-symbol at coordinate $(T+i)$ will be recovering precisely the erased symbol $c_{\ell_i}$. Thus in total, the set of symbols $\{c_T,c_{T+1},\ldots,c_{T+B-N-1}\}$ will correct $(B-N)$ erasures.

\item $\underline{0< \epsilon< (B-N)}$: Let $\zeta$ indicate the number of non-erased coordinates lying in $[T+\epsilon,T+\delta-1]$. Clearly, $\zeta>0$ iff $\epsilon<\delta$. Assuming $\zeta>0$, each such non-erased $c_{i+T}$ that lies in $\epsilon\leq i\leq \delta-1$, helps in recovering one unique erased symbol from the erasure burst. The argument, again, is similar to the one in the proof of Proposition \ref{prop:del_gt_B_minus_N}. This also means that the symbols $\{c_{aB+j}:\epsilon\leq j\leq\delta-1\}$ are known by time $T+\delta-1$.

Now let us see how non-erased symbols in the range $c_{T+i}:\max\{\delta,\epsilon\}\leq i\leq (B-N-1)$ help in the recovery of erased symbols. Note that each $c_{T+i}$ in the range $\max\{\delta,\epsilon\}\leq i\leq (B-N-1)$ takes the form: $c_{T+i}=c_i+c_{B+i}+\ldots+c_{(a-1)B+i}+\gamma_{i,0}c_{aB}+\gamma_{i,1}c_{aB+1}+\ldots+\gamma_{i,\delta-1}c_{aB+\delta-1}$. As we assume $\epsilon>0$, we can remove interference from the set of non-erased symbols $\{c_i,c_{B+i},\ldots,c_{(a-2)B+i}\}$ (these can be easily verified to be non-erased symbols, when $\epsilon>0$). For $\zeta>0$, we can remove the interference from the set of symbols $\{c_{aB+j}:\epsilon\leq j\leq\delta-1\}$ (these are known by time $T+\delta-1$). Thus effectively, each non-erased symbol $c_{T+i}$ in the range $\max\{\delta,\epsilon\}\leq i\leq (B-N-1)$ provides the sum $\hat{c}_{T+i}=c_{(a-1)B+i}+\gamma_{i,0}c_{aB}+\gamma_{i,1}c_{aB+1}+\ldots+\gamma_{i,min\{\epsilon,\delta\}-1}c_{aB+min\{\epsilon,\delta\}-1}$. There are the following two possibilities.


\bit 
\item $\underline{(B-\epsilon)\leq(N+\delta)}$: This condition essentially means that none of the $c_{(a-1)B+i}$ terms appearing in  $\hat{c}_{T+i}$, for $\max\{\delta,\epsilon\}\leq i\leq (B-N-1)$, are part of the burst erasure considered. In this scenario, we can further remove the interfering $c_{(a-1)B+i}$'s from $\hat{c}_{T+i}$'s, where $\max\{\delta,\epsilon\}\leq i\leq (B-N-1)$. Furthermore, if $(B-\epsilon)<\delta$, not all of the symbols $\{c_{aB},c_{aB+1},\ldots,c_{aB+\delta-1}\}$ will be erased. Hence we can rewrite the effective check-sum contributed by each $\hat{c}_{T+i}$ as 
$\hat{c}_{T+i}=\gamma_{i,\max\{\delta-B+\epsilon,0\}}c_{aB+\max\{\delta-B+\epsilon,0\}}+\gamma_{i,1}c_{aB+1}+\ldots+\gamma_{i,min\{\epsilon,\delta\}-1}c_{aB+min\{\epsilon,\delta\}-1}$. The number of remaining erased symbols $c_{aB+i}$ in the range $0\leq i\leq \min\{\epsilon,\delta\}-1$, $\{c_{aB+\max\{\delta-B+\epsilon,0\}}, c_{aB+\max\{\delta-B+\epsilon,0\}+1},$
$\ldots,$ $c_{aB+\min\{\epsilon,\delta\}-1}\}  $ and number of non-erased symbols in the coordinate range $[T+\max\{\epsilon,\delta\},T+B-N-1]$ is given in Table \ref{tab:case1}, for all the cases with $(B-\epsilon)\leq (N+\delta)$. From Table \ref{tab:case1}, for all the cases, the number of remaining erased symbols $\geq$ number of code-symbols (evaluations) in the range $[T+\max\{\epsilon,\delta\},T+B-N-1]$. As $\gamma_{x,y}$'s are chosen carefully as part of a Cauchy matrix, all these equations (evaluations) will be independent. Thus the number of independent equations provided by coordinates in the range $[T+\epsilon,T+B-N-1]$ is equal to $(B-N-\epsilon)$.	

	\begin{table}[ht]
	\centering
	\begin{tabular}{ |c|c|c|c|c|c| } 
		
		\hline
		Condition & No. of remaining erased (unknown)  & No. of  equations  \\ 
		& symbols among the coordinates & from the coordinates\\
		& $[aB,aB+\delta-1]$ & $[T+\max\{\epsilon,\delta\},T+B-N-1]$\\
		\hline
		$(B-\epsilon)\geq \delta$, $\epsilon\geq \delta$   &  $\delta$ &   $(B-N-\epsilon)$   \\
		&   &  \\
		\hline
		$(B-\epsilon)\geq \delta$, $\epsilon < \delta$   & $\epsilon$  &  $(B-N-\delta)$    \\
		&   &  \\
		\hline
		$(B-\epsilon)< \delta$,  $\epsilon \geq \delta$ & $(B-\epsilon)$  &  $(B-N-\epsilon)$    \\
		&   &  \\
		\hline
		$(B-\epsilon)< \delta$, 	$\epsilon < \delta$    & $(B-\delta)$  &  $(B-N-\delta)$    \\
		&   &  \\
		\hline
	\end{tabular}
	\caption{A summary of all the conditions, when $(B-\epsilon)\leq (N+\delta)$.}
	\label{tab:case1}
\end{table}

\item $\underline{(B-\epsilon)>(N+\delta)}$: Let $\alpha'\triangleq(B-\epsilon-N-\delta)$. The number of erased symbols recovered by the non-erased coordinates lying in $[T+\epsilon,T+\delta-1]$ is $\max\{\delta-\epsilon,0\}$.  Note that $\alpha'<(B-N-\delta)$. For each $c_{(a-1)B+i}$ such that $i\in [B-N-\alpha',B-N-1]$, there exists a non-erased $c_{T+i}$ which results in a $\hat{c}_{T+i}=c_{(a-1)B+i}+\gamma_{i,0}c_{aB}+\gamma_{i,1}c_{aB+1}+\ldots+\gamma_{i,\min\{\epsilon,\delta\}-1}c_{aB+\min\{\epsilon,\delta\}-1}$, after removing the known interfering symbols. To the contrary, if $c_{T+i}$ is part of the erased symbols, this would mean that $B\geq (B-N-i)+N+\delta+i\implies\delta \leq 0$, which is a contradiction. Thus $c_{T+i}$ for $i\in [B-N-\alpha',B-N-1]$ gives $\alpha'$ independent equations as they all give one unique component $c_{(a-1)B+i}$ each. For $j\in[\max\{\epsilon,\delta\},B-N-\alpha'-1]$, $c_{T+j}$'s are non-erased and results in a $\hat{c}_{T+j}=\gamma_{j,0}c_{aB}+\gamma_{j,1}c_{aB+1}+\ldots+\gamma_{j,\min\{\epsilon,\delta\}-1}c_{aB+\min\{\epsilon,\delta\}-1}$, after removing the known interfering symbols. These corresponds to $(B-N-\alpha'-\max\{\epsilon,\delta\})=\min\{\epsilon,\delta\}$ equations and $\min\{\epsilon,\delta\}$ unknowns and hence are all independent. Therefore, the total number of independent evaluations supplied by the non-erased coordinates in $[T+\epsilon,T+B-N-1]$ is $\max\{\delta-\epsilon,0\}+\alpha'+\min\{\epsilon,\delta\}=\delta+\alpha'=B-N-\epsilon$.
\eit
\een
Therefore for all the cases of $0\leq \epsilon\leq (B-N)$, $(B-N-\epsilon)$ non-erased code-symbols in the range $[T+\epsilon,T+B-N-1]$ provide that many evaluations, which are independent (over $\mathbb{F}_q$), with respect to the $T+1-(B-\epsilon)$ evaluation points corresponding to the $(T-B+\epsilon+1)$ non-erased locations in $[0,T-1]\cup\{T-N+B\}$. This essentially adds up to $T-N+1$ evaluations over independent evaluation points. As $k=T-N+1$, all the erasures can be corrected. 

In order to prove that the delay constraint of $T$ is met during the recovery from burst erasures, we need to consider only the burst erasures involving at least one of the coordinates $[0,B-N-1]$. It is enough to show that for any erasure burst of length $B$ that starts at coordinate $i\in[0,B-N-1]$, the code-symbol $c_i$ can be recovered with a delay of at most $T$.

 If $c_i$, where $i\in[0,\delta-1]$, is part of a burst-erasure, we fall back on the proof of the corresponding part in Proposition \ref{prop:del_gt_B_minus_N} to see that delay conditions are met. In the following we consider the case of $i\in[\delta,B-N-1]$. 

Consider an erasure burst which starts at $i$, where $i\in[\delta,B-N-1]$. The symbols $\{c_{i+B},c_{i+2B},$
$\ldots,c_{i+(a-1)B},c_{T+i}\}$ will be non-erased. Let $\epsilon$ be as defined previously. If $\epsilon=0$, all the code-symbols in $\{c_{aB},c_{aB+1},\ldots,c_{aB+\delta-1}\}$ will be known by the time $(T+i)$. This follows from the fact that all the symbols $c_T,c_{T+1},\ldots,c_{T+\delta-1}$ are non-erased. After removing the interference from known symbols at coordinates $\{i+B,i+2B,\ldots,i+(a-1)B,aB,aB+1,\ldots,aB+\delta-1\}$, the code-symbol $c_{T+i}$ will yield $c_i$.

For $\epsilon>0$, we first note that, for a burst starting at coordinate $i\in[\delta,B-N-1]$,  $a$ is forced to be equal to $1$. i.e., $T=B+\delta$. From the $\max\{\delta-\epsilon,0\}$ non-erased coordinates in $[T+\epsilon,T+\delta-1]$, we can obtain code-symbols in the coordinate range $[aB+\epsilon,aB+\delta-1=T-1]$. Note that $i=T-B+\epsilon=\delta+\epsilon$. Number of non-erased coordinates in the range $[T+\max\{\delta,\epsilon\},T+i]=i-\max\{\delta,\epsilon\}+1$. For each $j\in [\max\{\delta,\epsilon\},i-1]$, $c_j$ is a non-erased symbol. This follows from our assumption that, the error burst is starting at time $i$. For $j\in [\max\{\delta,\epsilon\},i]$, $c_{T+j}\triangleq c_j+\gamma_{j,0}c_{B}+\gamma_{j,1}c_{B+1}+\ldots+\gamma_{j,\delta-1}c_{B+\delta-1}$ (as $a=1$). After removing the known symbols from the coordinate range $[B+\epsilon,B+\delta-1]$ and $[\max\{\delta,\epsilon\},i-1]$, we have the resultant sums:

\bean
\hat{c}_{T+j} & = & \left\{ \begin{array}{rl} \gamma_{j,0}c_{B}+\gamma_{j,1}c_{B+1}+\ldots+\gamma_{j,\min\{\epsilon,\delta\}-1}c_{B+\min\{\epsilon,\delta\}-1} & \text{if $\max\{\delta,\epsilon\}\leq j\leq i-1$}, \\
	c_j+\gamma_{j,0}c_{B}+\gamma_{j,1}c_{B+1}+\ldots+\gamma_{j,\min\{\epsilon,\delta\}-1}c_{B+\min\{\epsilon,\delta\}-1} & j=i. \end{array} \right.
\eean

For the $i-\max\{\delta,\epsilon\}=\delta+\epsilon-\max\{\delta,\epsilon\}=\min\{\delta,\epsilon\}$ equations in the range $\max\{\delta,\epsilon\}\leq j\leq i-1$, there are that many unknowns. As $\gamma_{x,y}$'s are carefully chosen to be part of a Cauchy matrix, these equations can be solved (at time $(T+i-1)$) to obtain the code-symbols: $\{c_{B}, c_{B+1},\ldots,c_{B+\min\{\epsilon,\delta\}}\}$. Thus at time $(T+i)$, the contribution from these code-symbols can be removed from $\hat{c}_{T+i}$ to obtain $c_i$. This completes the proof.

\bibliographystyle{IEEEtran}
\bibliography{codes_for_streaming}

\end{document}